\documentclass[a4paper,11pt]{article}

\usepackage{amsfonts,amssymb,amsmath,amscd,epic,graphics}
\usepackage{psfrag,epsfig}
\usepackage{graphicx}
\usepackage{amsfonts,color}  
\usepackage{subfigure}
\usepackage{psfrag,xspace}
\usepackage{multirow} 

\usepackage[naturalnames]{hyperref} 
\usepackage{algorithmic} 
\usepackage[boxruled,boxed,algo2e]{algorithm2e}
\usepackage{ifthen}
\usepackage{lscape}
\usepackage{marvosym}

\usepackage{vector}
\usepackage{enumerate} 

\usepackage{amsfonts}

\usepackage{fullpage}
\usepackage{notations}

\usepackage{hyperref} 
\usepackage{framed}
\newtheorem{theorem}{Theorem}

\newtheorem{proposition}[theorem]{Proposition}
\newtheorem{corollary}[theorem]{Corollary}
\newtheorem{remark}[theorem]{Remark}
\newtheorem{example}[theorem]{Example}

\title{The DMM bound:\\
Multivariate (aggregate) separation bounds}

\author{
  Ioannis Z. Emiris%
  \thanks{National Kapodistrian University of Athens, Greece.
    Email: \texttt{emiris(at)di.uoa.gr}}
  \and 
  Bernard Mourrain%
  \thanks{INRIA M\'editerran\'ee, Sophia-Antipolis, France.
    Emai: \texttt{mourrain(at)inria.fr}}
  \and 
   Elias P.~Tsigaridas%
   \thanks{\AA{}rhus University, Denmark.
     Email: \texttt{elias.tsigaridas(at)gmail.com}}
}

\newcommand{\MV}{\ensuremath{\mathsf{MV}}\xspace}
\newcommand{\M}{\ensuremath{\mathsf{M}}\xspace}
\newcommand{\dmm}{\texttt{DMM}\xspace}
\renewcommand{\sep}{\ensuremath{\mathsf{sep}}\xspace}

\begin{document}
\maketitle

\begin{abstract}
In this paper we derive aggregate separation bounds, named after
Davenport-Mahler-Mignotte (\dmm),
on the isolated roots of polynomial systems, specifically
on the minimum distance between any two such roots.
The bounds exploit the structure of the system
and the height of the sparse (or toric) resultant by means of mixed volume,
as well as recent advances on aggregate root bounds for univariate
polynomials,
and are applicable to arbitrary positive dimensional systems.
We improve upon Canny's gap theorem \cite{c-crmp-87}
by a factor of $\OO(d^{n-1})$, where $d$ bounds the degree
of the polynomials, and $n$ is the number of variables.
One application is to the bitsize of the eigenvalues and eigenvectors
of an integer matrix, which also yields a new proof that the problem
is polynomial.
We also compare against recent lower bounds on the absolute value
of the root coordinates by Brownawell and Yap \cite{by-issac-2009}, 
obtained under the hypothesis there is a 0-dimensional projection.
Our bounds are in general comparable, but exploit sparseness;
they are also tighter when bounding the value of a positive polynomial over the simplex.
For this problem, we also improve upon the bounds in \cite{bsr-arxix-2009,jp-arxiv-2009}.
Our analysis provides a precise asymptotic upper bound on the number of steps that
subdivision-based algorithms perform in order to isolate all real roots of a polynomial system.
This leads to the first
complexity bound of Milne's algorithm \cite{Miln92} in 2D.
\end{abstract}

\section{Introduction}

One of the great challenges in algebraic algorithms is to fully understand
the theoretical and practical complexity of methods based on exact arithmetic.
One goal may be towards 
hybrid {symbolic-numeric}
approaches that exploit both exact and approximate computations.
Computing all roots, 
in some representation, of systems of multivariate polynomials
is a fundamental question in both symbolic and numeric computation. 
The complexity analysis and the actual runtimes typically
depend on {\em separation bounds}, i.e.\ the minimum
distance between any two, possibly complex, roots of the system.
This is particularly true for algorithms based on subdivision techniques
and, more generally, for any numerical solver seeking to certify its output.
Hence, separation bounds are of great use in areas
such as computational geometry and geometric modeling.

Davenport \cite{Dav:TR:85} was first to introduce aggregate separation bounds
for the real roots of a univariate polynomial, which depend on Mahler's measure,
e.g.\ \cite{Mign91}.
Mignotte \cite{Mignotte:AAECC:95} loosened the hypothesis on the
bounds and extended them to complex roots. 

As for algebraic systems, a fundamental result is Canny's Gap theorem \cite{c-crmp-87}, 
on the separation bound for square 0-dimensional systems, see Th.~\ref{th:gap}.
Yap \cite{Yap2000} relaxed the 0-dimensional requirement by
requiring it holds only on the affine part of the variety.
A recent lower bound on the absolute value of the root coordinates
\cite{by-issac-2009} applies to those coordinates for which the
variety's projection has dimension~0, and does not require the system
to be square.
For arithmetic bounds applied to Nullstellensatz we refer to \cite{kps-djm-2001}.

Basu, Leroy, and Roy \cite{bsr-arxix-2009} and, recently, Jeronimo and Perrucci
\cite{jp-arxiv-2009} considered the closely related problem of
computing a lower bound for the minimum value of a positive polynomial
over the standard simplex. For this, they compute lower bounds on the
roots of polynomial system formed by the polynomial and all its partial derivatives.
This problem is also treated in \cite{by-issac-2009}.

Separation bounds are important for estimating the complexity of
subdivision-based algorithms for solving  polynomial systems,
that depend on exclusion/inclusion predicates 
or root counting techniques,
e.g.~\cite{yakoubsohn-joc-2005,MANTZAFLARIS:2009:INRIA-00387399:1,bcgy-issac-2008,Miln92,gvt-msturm-1997}.

\paragraph {\bf Our contribution.}
We derive worst-case
(aggregate) separation bounds for the roots of polynomial systems,
which are not necessarily 0-dimensional.
The bounds are computed as a function of the number of variables,
the norm of the polynomials, and a bound on the number of roots
of well-constrained systems.
For the latter we employ mixed volume in order to exploit the sparse
structure that appears in many applications.
Any future better bound can be used to improve our results.
The main ingredients of our proof are resultants, including bounds on
their height \cite{sombra-ajm-2004}.

We extend the known separation bound for single polynomial equations
to 0-dimensional systems,
and call it $\dmm_n$, after {\em Davenport-Mahler-Mignotte}.
This improves upon Canny's Gap theorem by $\OO(d^{n-1})$.
Our bounds are within a factor of $\OO(2^n)$ from optimal for certain
systems, which is good for $n$ small (or constant) compared to the other parameters.
They are comparable to those in \cite{by-issac-2009} on the absolute value
of root coordinates, but they are an improvement when expressed using mixed volumes.
It seems nontrivial to apply sparse elimination theory
to the approach of \cite{by-issac-2009}.
More importantly, our result is extended to positive-dimensional systems,
thus addressing a problem that has only been examined very recently in \cite{by-issac-2009}.

We illustrate our bounds on computing the eigenvalues /
eigenvectors of an integer matrix, 
and improve upon Canny's bound by a factor exponential in matrix dimension.
Thanks to mixed volume, we derive a bound
polynomial in the logarithm of the input size, hence offering
a new alternative to Bareiss' result \cite{bareiss-moc-1968}
that the problem is of polynomial bit complexity. 
We also bound the minimum of a positive polynomial over the standard simplex
and improve upon the 3 best known bounds
\cite{bsr-arxix-2009,by-issac-2009,jp-arxiv-2009}, 
when the total degree is larger than the number of variables.

Finally, we upper bound the number of steps for any subdivision based algorithm
using a real-root counter in a box to isolate the real roots
of a system in a given domain. 
This leads to the first complexity bound of Milne's algorithm \cite{Miln92} in $\RR^2$.
This aggregate separation bound is also useful in the analysis of
the subdivision algorithm based on continued fractions expansion
\cite{MANTZAFLARIS:2009:INRIA-00387399:1} for polynomial system solving.

The polynomial systems in practice have a small number of real roots
and all roots, real and complex, are well separated;
it is challenging to derive an average-case $\dmm_n$.
Another open question is to express the positive-dimensional bound wrt
the dimension of the excess component.  

{\bf Paper structure.}
We introduce some notation, then Sec.~\ref{sec:DMM} derives and proves
the multivariate version of $\dmm$ as main Thm.~\ref{th:DMM-n}.
Its near-optimality and comparisons to existing bounds are in Sec.~\ref{sec:opt-ext},
which also 
extends it to positive-dimensional systems.
Two applications of our bounds are in Sec.~\ref{sec:applications}.
Sec.~\ref{sec:subdiv} is devoted to subdivision algorithms.

\paragraph{ Notation.} 

\OO, resp.\ \OB, means bit, resp.\ arithmetic,
complexity and \sOB, resp.\ \sO, means we are ignoring logarithmic factors.
For a polynomial $f \in$ $\ZZ[x_1, \dots, x_n]$, where $n \geq 1$,
$\deg( f)$ denotes its total degree, 
while $\deg_{x_i}(f)$ denotes its degree w.r.t.~$x_i$.
By \bitsize{f} we denote the maximum bitsize of the coefficients of $f$
(including a bit for the sign).  For $a \in \QQ$, $\bitsize{ a} \ge 1$
is the maximum bitsize of the numerator and denominator.
For simplicity, we assume,
for any polynomial, $\log(\dg{f}) = \OO( \bitsize{f})$.
Let $\sep( f)$, resp.\ $\sep( \Sigma)$, denote the separation bound, i.e.\
the minimum distance between two, possibly complex, roots of polynomial
$f$, resp.\ system $(\Sigma)$.  
For $f=a_{d}\prod_{i=1}^{d}(x-z_{i})$ $\in \mathbb{C}[x]$, with $a_{d} \neq 0$,
its Mahler measure is $\Mahler{f}:=$ $4 |a_{d}| \prod_{i=1}^{d} \max\{1, |z_{i}|\}$.

\section{The \texttt{DMM} bound}
\label{sec:DMM}

{\bf The univariate case.}
Consider a real univariate polynomial $A$, not necessarily square-free,
of degree $d$ and its complex roots $\gamma_j$
in ascending magnitude, where $j \in \{ 1, 2, \dots, d\}$.
The next theorem \cite{te-tcs-2008}
bounds the product of differences of the form $|\gamma_i - \gamma_j|$.
It slightly generalizes a theorem in \cite{Mign91},
which in turn generalizes \cite{Dav:TR:85},
see also \cite{Johnson-phd-91,ESY:descartes}.
\begin{theorem}[$\dmm_1$]
  \label{th:DMM-1}
  Let $f \in \CC[X]$, with $\deg(f) = d$ and 
  not necessarily square-free.
  Let $\Omega$ be any set of $\ell$ couples of indices $(i, j)$,
  $1 \leq i < j \leq d$,
  and let the distinct non-zero (complex) roots of  $f$ be
  $0 < | \gamma_1 | \leq |\gamma_2| \leq \dots \leq |\gamma_d|$.
  Then
  {
  \begin{displaymath}
    2^{\ell} \mathcal{M}(f)^{\ell}\geq 
    \prod_{(i, j) \in \Omega}{|\gamma_i - \gamma_j|}\geq
    2^{\ell - \frac{d(d-1)}{2}} 
    {\mathcal{M}(f)}^{1-d - \ell}
    \sqrt{ \abs{ \disc(f_{red})}},
  \end{displaymath}  
  }
  where $f_{red}$ is the square-free part of $f$.
  If $f \in \ZZ[x]$, $\ell \leq d$ and $\bitsize{f} = \tau$, then 
  {
  \begin{displaymath}
    d^{d/2} \, 2^{2 d \tau}   \geq
    d^{\ell/2} \, 2^{2 \ell\tau}   \geq 
    \prod_{(i, j) \in \Omega}{|\gamma_i - \gamma_j|}  
    \geq
    d^{-d} \, 2^{-d^2 - 3\tau(\ell + d)} \geq
    d^{-d} \, 2^{-d^2 - 6 d \tau}.
  \end{displaymath}
  }
\end{theorem}
The second inequality follows from:
$\Mahler{ f} \leq \Norm{ f}_2 \leq (d+1) \norm{f}_{\infty} \leq (d+1)^{\frac{1}{2}} \, 2^{\tau}$,
e.g.\ \cite{Mign91,Yap2000}.
In the first inequality we can replace $\Mahler{f}$ by $\norm{f}_2$.

The bound of Thm.~\ref{th:DMM-1}
has an additional factor of $2^{d^2}$ wrt \cite{Dav:TR:85,ESY:descartes}, 
which is, asymptotically, not significant when the polynomial is not square-free 
or $d = \OO(\tau)$.
The current version of the theorem has very loose hypotheses and 
applies to non-squarefree polynomials.

Roughly, $\dmm_1$ provides a bound on all distances between
consecutive roots of a polynomial. This quantity is, asymptotically, almost equal
to the separation bound.
The interpretation is that not all roots of a polynomial
can be very close together or, quoting J.H.\ Davenport,
{\em ``not all [distances between the roots] could be bad''}.  

{\bf The multivariate case.}
This section generalizes $\dmm_1$ to 0-dimensional polynomial systems.
Let $n > 1$ be the number of variables.
We use $\bold{ x}^{\bold{ e}}$ to denote the monomial
$x_1^{e_1} \cdots x_n^{e_n}$, with
$\bold{ e} = (e_1, \dots, e_n) \in \ZZ^n$.
The input is {\em Laurent polynomials}
$f_1, \dots, f_n \in K[x_1^{\pm}, \dots, x_n^{\pm}]$
$= K[ \bold{ x}, \bold{ x}^{-1}]$, where $K\subset \mathbb{C}$ is the coefficient field.
Since we can multiply Laurent polynomials by monomials without affecting
their nonzero roots, in the sequel we assume there are no negative exponents.
Let the polynomials be 
\begin{equation}
  f_i = \sum_{j=1}^{m_i}{ c_{i, j} \bold{ x}^{a_{i, j}}}, \quad 1 \leq i \leq n.
  \label{eq:polys}
\end{equation}
Let $\Set{ a_{i, 1}, \dots, a_{i, m_i}} \subset \ZZ^n$ be the support of $f_i$;
its Newton polytope $Q_i$ is the convex hull of the support.
Let $\MV( Q_1, \dots, Q_n)>0$ be the {\em mixed volume} of
convex polytopes $Q_1, \dots, Q_n \subset \RR^n$.
Here is Bernstein's bound, known also as BKK bound.

\begin{theorem}
  \label{th:bkk}
  For $f_1, \dots f_n \in \CC[ \bold{ x}, \bold{ x}^{-1}]$ with Newton polytopes
  $Q_1, \dots, Q_n$, the number of common isolated solutions in $( \CC^{*})^n$,
  multiplicities counted, does not exceed $\MV( Q_1, \ldots$, $Q_n)$,
  independently of the corresponding variety's dimension.
\end{theorem}

We consider polynomial system
\begin{equation}
  \label{eq:system}
(\Sigma): \; f_1(\bold{x}) = f_2(\bold{x}) = \cdots = f_n(\bold{x}) = 0,
\end{equation}
where $f_i \in \RR[\bold{x}^{\pm 1}]$, which we assume to be 0-dimensional.
We are interested in its roots in $(\CC^{*})^n$, which are called toric.
We denote by $Q_{0}$ the convex hull of the unit standard simplex.
Let $\M_{i} = \MV(Q_{0}, \dots, Q_{i-1}, Q_{i+1}, \dots, Q_n)$, and
$\#Q_i$ denote the number of lattice points in the closed polytope $Q_i$.
Wlog, assume $\dim\sum_{i=0}^n Q_i=n$ and $\dim{\sum_{i\in I} Q_i} \ge j$
for any $I\subset\{0,\dots,n\}$ with $|I|=j$.
We consider the {\em sparse (or toric) resultant} of a system of
$n+1$ polynomial equations in $n$ variables, assuming we have fixed the $n+1$ supports.
It provides a condition on the coefficients for the solvability of the system, and
generalizes the classical resultant of $n$ homogeneous polynomials, by taking into account the
supports of the polynomials.
For details, see \cite{CLO2}.

Let $D$ be the number of roots $\in (\CC^{*})^n$ of
$(\Sigma)$, multiplicities counted, so $D \leq \M_{0}$. 
We also use 
$B = (n-1)\, {D \choose 2}$, and $\dg{f_i} = d_i \leq d$.
If $f_i \in \ZZ[ \bvec{x}^{\pm 1}]$, 
$\bitsize{f_i} = \tau_{i} \leq \tau$, $1\leq i \leq n$.
Now $\vol( \cdot)$ stands for Euclidean volume, and
$(\#Q_i)$ for the number of lattice points in $Q_i$;
the inequality connecting $(\#Q_i)$ and polytope volume
is in \cite{blichfeldt-tams-1914}.
We present the abbreviations and inequalities used throughout the paper:

{ 
\begin{eqnarray}
  \begin{aligned}
    & D \leq M_0 \leq \prod_{i=1}^{n}{d_i} \leq d^n,  
    && B \leq n D^2 \leq n  \prod_{i=1}^{n}{d_i^2} \leq nd^{2n},  \\ 
    & \M_i \leq  \prod_{1 \leq j \leq n \atop j\not= i}{d_j} = D_i,
    && \sum_{i=1}^{n}{\M_i} \leq \sum_{i=1}^{n}{D_i} \leq nd^{n-1}, \\
  \end{aligned}  \\ \nonumber
  \begin{aligned}
    & (\# Q_i) \leq n! \vol(Q_i) + n \leq d_i^n + n \leq 2d_i^n, \\
    &A = \prod_{i=1}^{n} \sqrt{\M_i} \, 2^{\M_i} \leq 2^{n d^{n-1} + 
      \frac{n^2-n}{2}\lg{d}}, \\
    &C = \prod_{i=1}^{n}{\norm{f_i}_{\infty}^{M_j}} \leq 2^{\tau \sum_{i=1}^{n}{ \M_i}}  \leq 2^{n \tau d^{n-1}},\\
    &h \leq (n+1)^D \varrho \leq (n+1)^{d^n} 2^{nd^{n-1}} d^{n^2d^{n-1}}, \\
    & \varrho = \prod_{i=1}^{n}{(\# Q_i)^{M_i}} \leq 2 ^{\sum_{i=1}^{n}{D_i}} \prod_{i=1}^{n}{d_i^{nD_i}} \leq
    2^{nd^{n-1}} \, d^{n^2d^{n-1}} .
  \end{aligned}
  \label{eq:all-def}
\end{eqnarray}
}

\begin{framed}
  \begin{theorem}[$\texttt{DMM}_n$]
    \label{th:DMM-n}
    Consider the 0-dimensional polynomial system $(\Sigma)$ in (\ref{eq:system}).
    Let $D$ be the number of complex solutions of the system in
    $(\CC^*)^n$, which are $0 < | \gamma_1 | \leq |\gamma_2| \leq
    \dots \leq |\gamma_D|$.  Let $\Omega$ be any set of $\ell$ couples
    of indices $(i, j)$ such that $1 \leq i < j \leq D$ 
    and $\gamma_i \not= \gamma_j$.
    Then the following holds
    {
      \begin{equation}       
        (2^{D+1} \, \varrho \, C)^{\ell} 
        \geq
        \prod_{(i, j) \in \Omega}{ \Abs{ \gamma_i - \gamma_j }} 
        \geq 
        2^{-\ell - (D-1)(D+2)/2} \, (h \,C)^{1-D-\ell} \, B^{(1-n)(D^2 + D(\ell-1) + \ell)} \, 
        \sqrt{ |U_{r}|},
        \label{eq:dmm-D}
      \end{equation}
    }
    where $| U_{r}|$ denotes the discriminant of the
    square-free part of the $u-$resultant, and $| \cdot |$ denotes absolute value.
    If  $f_i \in \ZZ[ \bvec{x}]$
    and $\gamma_{j,k}$ stands for the $k$-th coordinate,
    $1 \leq k  \leq n$, of  $\gamma_j$, then:
    \begin{equation}
      (2^{D} \, \varrho \, C)^{-1}
      \leq  |\gamma_{j,k} |  \leq 
      2^{D} \, \varrho \, C,  \label{eq:upper-lower-D}
    \end{equation}  
    \begin{equation}
      \sep( \Sigma) 
      \geq 
      2^{-(3D+2)(D-1)/2} \,( \sqrt{D+1} \, \varrho \, C)^{-D}. \label{eq:sep-D}
    \end{equation}
  \end{theorem}
\end{framed}

The following corollary employs mixed volumes.
\begin{corollary}\label{CorMVs}
Under the hypothesis of Th.~\ref{th:DMM-n}, for $f_{i} \in
\ZZ[\mathbf{x}^{\pm 1}]$, $i=1,\ldots,n$, we have
\begin{eqnarray*}
  \lefteqn{      2^{\M_{0} \, (1 + \M_{0} + \sum_{i=1}^{n} \M_{i} \, (\tau +\lg(\#Q_i) ))} } \\
  &\geq \prod_{(i, j) \in \Omega}{ \Abs{ \gamma_i - \gamma_j }} \geq \\
  &2^{
    -2\,\M_{0} \, \sum_{i=1}^{n} \M_{i} \, (\tau +\lg(\#Q_i)  ) \, 
    -2 \,\M_{0}^{2}  \, (1 + \, \lg(n+1)+ n \, \lg(n) + 2\, n \, \lg(\M_{0}))},
  \label{eq:dmm-dt0}
\end{eqnarray*}
\begin{equation}
  2^{-( \M_{0} + \sum_{i=1}^{n} \M_{i} \, (\tau +\lg(\#Q_i) ))}
  \leq |\gamma_{j,k} | \leq 
  2^{\M_{0} + \sum_{i=1}^{n} \M_{i} \, (\tau +\lg(\#Q_i) )}
\end{equation}
\begin{equation}
  \mathsf{sep}( \Sigma) \geq 
  2^{- \M_{0}({3\over 2} \M_{0} + \lg(\M_{0}) +\, \sum_{i=1}^{n} \M_{i} \, (\tau +\lg(\#Q_i)  ))}.
\end{equation}
\end{corollary}

\begin{corollary}
Under the hypothesis of Th.~\ref{th:DMM-n}, for $f_{i} \in
\ZZ[\mathbf{x}^{\pm 1}]$, $\dg{f_{i}} \leq d$ and $\bitsize{f_i} \leq \tau$, 
we have
  {
    \begin{equation}
      \prod_{(i, j) \in \Omega}{ \Abs{ \gamma_i - \gamma_j }}
      \geq
      2^{-(3+4\lg{n}+4n\lg{d})d^{2n}} \, 2^{-2n(1+n\lg{d}+\tau)d^{2n-1}},
      \label{eq:dmm-dt}
    \end{equation}
    \begin{equation}
      2^{-d^n - n(\tau + n\lg{d} + 1)d^{n-1}}
      \leq |\gamma_{j,k} | \leq 
      2^{d^n + n(\tau + n\lg{d} + 1)d^{n-1}}, \label{eq:upper-lower-dt}
    \end{equation}
    \begin{equation}
      \mathsf{sep}( \Sigma) \geq 
      2^{-2d^{2n} -n(2n\lg{d} + \tau)d^{2n-1}}.
      \label{eq:sep-dt}
    \end{equation}
  }
\end{corollary}

{\bf Proof of main theorem.}
Let us first establish the lower bound.
Let $\gamma_i = (\gamma_{i,1}, \dots, \gamma_{i, n})
\in (\CC^{*})^n$, $1 \leq i \leq D$, be the solutions 
of $(\Sigma)$, where $f_i$ are defined in (\ref{eq:polys}),
We denote the set of solutions as $V \subset (\CC^{*})^n$.
We add an equation to $(\Sigma)$ to obtain:
\begin{equation}
  \label{eq:over-system}
(\Sigma_0):\; f_0( \bold{x}) = f_1( \bold{x}) = \cdots = f_n( \bold{x}) = 0,
\end{equation}
where
\begin{equation}
  \label{eq:u-poly}
  f_0 = u + r_1 x_1 + r_2 x_2 + \cdots + r_n x_n,
\end{equation}
$r_1, \dots, r_n\in\ZZ$ to be defined in the sequel, and $u$ is a new parameter.
Now $u= -\sum_{i}{ r_i \, \gamma_{j, i}}$, on a solution $\gamma_j$.
We choose properly the coefficients of $f_0$ to ensure that the function
\begin{displaymath}
    f_0 : V  \rightarrow \CC^{*} \, :\, \gamma \mapsto  f_0( \gamma)
\end{displaymath}
is injective.
The separating element shall ensure injectivity
\cite{BPR06,c-crmp-87,EmiPhd,Rou:rur:99}.
\begin{proposition}
  \label{prop:sep-element}
  Let $V \subset \CC^n$ with cardinality $D$. The set of linear forms
  \begin{displaymath}
    \mathcal{F} = \set{ u_i = x_1 + i \,x_2 + \dots + i^{n-1} \,x_n \,|\, 0 \leq i \leq B = (n-1){ D \choose 2} }
  \end{displaymath}
  contains at least one {\em separating element}, which takes
  distinct values on $V$. 
\end{proposition}
\begin{corollary}
  \label{cor:sep-elem-bitsize}
  For $f_0 \in \mathcal{F}$ it holds that 
  $\norm{f_0}_{\infty} \leq B^{n-1}$, and
  \begin{displaymath}
   \norm{ f_0}_{\infty} \leq 
  \norm{ f_0}_2 \leq 2 B^{n-1} = 2 (n-1)^{n-1} \, {D \choose 2}^{n-1}.
\end{displaymath}

\end{corollary}
\begin{proof}
  The first inequality is evident from the definition of infinite norm.
  For the second inequality, $B = (n-1)\, {D \choose 2}$:
  \begin{displaymath}
    \begin{array}{lll}
      \norm{ f_0}_{\infty} & \leq 
      \norm{ f_0}_2   
      \leq \sqrt{1 + B^2 + B^4 + \cdots + (B^{2})^{n-1}} \\
      & \leq  \sqrt{ \frac{B^{2n}-1}{B^2 - 1}} 
      \leq \sqrt{ \frac{B^{2n-2}}{1 - 1/B^2}}
      \leq \sqrt{ 4 B^{2n-2}} = 2 B^{n-1}.
    \end{array}
  \end{displaymath}
\end{proof}

We consider the {\em $u-$resultant} $U$ of $(\Sigma_0)$ that eliminates $\bold{x}$.
It is univariate in $u$, with coefficients homogeneous polynomials
in the coefficients of $(\Sigma_0)$, e.g.~\cite{CLO2}:
\begin{equation}
  U(u) = \dots + 
  {\varrho}_k \, u^k  \, \bold{r}_k^{D - k} \bold{c}_{1,k}^{\M_{1}} \bold{c}_{2,k}^{\M_{2}} \cdots \bold{c}_{n,k}^{\M_{n}}
  + \dots,
  \label{eq:U-monom}
\end{equation}
where $\varrho_k \in \ZZ$,
$\bold{c}_{j,k}^{\M_{j}}$  denotes a monomial in coefficients of
$f_j$ with total degree $\M_{j}$,
and $\bold{r}_k^{D- k}$ denotes a monomial in the coefficients of $f_0$ 
of total degree $D - k$.
The degree of $U$, with respect to $u$ is $D$.  It holds that 
\begin{equation}
  \label{eq:c-bound}
  \Abs{ \bold{c}_{1,k}^{\M_1} \,\bold{c}_{2,k}^{\M_2}\, \dots \, \bold{c}_{n,k}^{\M_n}}
  \leq C = \prod_{i=1}^{n}{ \norm{ f_i}_{\infty}^{\M_i}},
\end{equation}

From Cor.~\ref{cor:sep-elem-bitsize}
we have that $\abs{ \bold{ r_k}} \leq \norm{ f_0}_{\infty} \leq B^{n-1}$,
for all $k$.
Let $|\varrho_k| \leq h$, for all $k$. Then using
\cite{sombra-ajm-2004}, see also Eq.~(\ref{eq:all-def}), 
we get that
\begin{displaymath}
  h \leq \Prod_{i=0}^{n}{ (\#Q_i)^{\M_i}} 
  = (\#Q_0)^{D}  \Prod_{i=1}^{n}{ (\#Q_i)^{\M_i}} 
  = (n+1)^{D} \varrho.
\end{displaymath}

We can bound the norm of $U$:
\begin{displaymath}
  \begin{array}{rclr}
    \norm{ U}_2^2 & \leq & \Sum_{k=0}^{D}{ \Abs{ \varrho_k \, \bold{r}_k^{D - k}
        \bold{c}_{1,k}^{\M_{1}} \bold{c}_{2,k}^{\M_{2}} \dots \bold{c}_{n,k}^{\M_{n}}}^2}  & \\
    & \leq &  \Sum_{k=0}^{D}{ \Abs{ h \, (B^{n-1})^{D - k} \, C}^2} 
    \leq h^2 \, C^2 \,  \Sum_{k=0}^{D}{ (B^{2n-2})^{D - k}} & \\
    & \leq & h^2 \, C^2 \,  \Sum_{k=0}^{D}{ (B^{2n-2})^{k}} 
    \leq  h^2 \, C^2 \,  \frac{ (B^{2n-2})^{D+1} - 1}{ B^{2n-2} - 1} & \\
    & \leq & h^2 \, C^2 \,  4 \, (B^{2n-2})^{D} \leq 4 \, h^2 \, C^2 \,  B^{2(n-1)D}, & 
  \end{array}
\end{displaymath}
and so
\begin{displaymath}
  \label{eq:u-resultant-norm}
  \norm{ U}_{\infty}  \leq \norm{ U}_2 
  \leq  2 \, h \, C \,  B^{(n-1)D}
  \leq 2 \, (n+1)^D \, \varrho \, C \,  B^{(n-1)D}.
\end{displaymath}

If $u_j$ are the distinct roots of $U$, then by recalling the
injective nature of $f_0$, we deduce that $u_j = -\sum_{i=1}^{n}{ r_i \, \gamma_{j, i}}$.
Actually the $u-$resultant is even stronger, since the multiplicities of
its roots correspond to the multiplicities of the solutions of the system,
but we will not exploit this further. 

\begin{proposition}[Cauchy-Bunyakovsky-Schwartz]
  \label{prop:cbs-ineq}
  Let $a_1$, $a_2$, $\dots, a_n \in \CC$ and $b_1, b_2, \dots, b_n \in \CC$. 
  Then,
  \begin{displaymath}
    \Abs{ \bar{ a}_1 \, b_1 + \dots + \bar{ a}_n \, b_n}^2 \leq 
    \Paren{ \abs{ a_1}^2 + \dots + \abs{ a_n}^2 } \Paren{ \abs{ b_1}^2 + \dots + \abs{ b_n}^2 },
  \end{displaymath}
  where $\bar{ a}_i$ denotes the complex conjugate of $a_i$, and $1 \leq i \leq n$.
  Equality holds if, for all $i$, $a_i = 0$
  or if there is a scalar  $\lambda$ such that $b_i = \lambda \, a_i$.
\end{proposition}

Consider $\gamma_i, \gamma_j$
and let $u_i, u_j$ be the corresponding roots of $U$.
Using Prop.~\ref{prop:cbs-ineq},
\begin{displaymath}
  \Abs{ r_1( \gamma_{i,1} - \gamma_{j, 1}) + \dots + r_n( \gamma_{i,n} - \gamma_{j, n}) }^2 \leq   
  \Paren{ r_1^2 + \dots + r_n^2}^2 
    \Paren{ \abs{ \gamma_{i,1} - \gamma_{j, 1}}^2 + \dots + \abs{ \gamma_{i,n} - \gamma_{j,n}}^2} 
    \Leftrightarrow
\end{displaymath}
\begin{displaymath}
    \Abs{ \sum_{k=1}^{n}{ r_k \, \gamma_{i, k}} -  \sum_{k=1}^{n}{ r_k \, \gamma_{j, k}}}^2
    \leq  \sum_{k=1}^{n}{ r_k^2} \cdot \sum_{k=1}^{n}{ \Abs{ \gamma_{i, k} - \gamma_{j, k}}}^2 
    \Leftrightarrow \Abs{ u_i - u_j}^2
    \leq 
    \Paren{ \sum_{k=1}^{n}{ r_k^2}} \cdot \abs{ \gamma_i - \gamma_j }^2,
\end{displaymath}

and thus
\begin{displaymath}
   \abs{ \gamma_i - \gamma_j } \geq  \Paren{ \sum_{k=1}^{n}{ r_k^2}}^{-1/2} \, \Abs{ u_i - u_j}.
\end{displaymath}

To prove the lower bound of Th.~\ref{th:DMM-n},
we apply the previous inequality for all pairs in $\Omega$, $\abs{ \Omega} = \ell$.
So we get 
 \begin{equation}
   \prod_{(i, j) \in \Omega}{ \Abs{ \gamma_i - \gamma_j}} 
   \geq 
   \Paren{ \sum_{k=1}^{n}{ r_k^2}}^{ -\frac{1}{2} \ell}
    \prod_{(i, j) \in \Omega}{ \Abs{ u_i - u_j }}.
    \label{eq:gamma-u}
\end{equation}

It remains to bound the two factors of RHS of the previous inequality.
To bound the first we use Cor.~\ref{cor:sep-elem-bitsize}.
It holds
$ 
\sum_{k=1}^{n}{ r_k^2} \leq 1 + \sum_{k=1}^{n}{ r_k^2}  
\leq \norm{ f_0}_2^2 \leq 4 \, B^{2n-2},
$
so 
\begin{equation}
  ( \sum_{k=1}^{n}{ r_k^2} )^{ -\frac{1}{2} \ell} 
  \geq 2^{-\ell} \, B^{(1-n)\ell}.
    \label{eq:rk-bound}
\end{equation}

For the second factor of (\ref{eq:gamma-u})
we apply $\texttt{DMM}_1$ to $U$; and thus

{
\begin{equation}
  \begin{array}{rl}
    \prod_{(i, j) \in \Omega}{ \Abs{ u_i - u_j }} \geq &
     2^{\ell -D(D-1)/2} \, \norm{U}_2^{1-D-\ell} \, \sqrt{| U_{r}|} \\
     \geq &2^{(1-D)(D +2)/2} \, (h \, C \, B^{(n-1)D})^{1-D-\ell} \, \sqrt{ |U_{r}|}.
  \end{array}
  \label{eq:u-bound}  
\end{equation}
}
Combining (\ref{eq:gamma-u}) with (\ref{eq:rk-bound}) and (\ref{eq:u-bound}), we have
the lower bound.
In the case where the polynomials are in $\ZZ[\bold{x}]$, then it holds that 
the absolute value of the discriminant of a square-free polynomial is $\geq 1$,
and we can omit it from the inequality.
If the polynomials are in $\QQ[ \bold{x}]$ the bounds are almost the same,
since they depend on Mahler's measure.  


Let us now establish the upper bound.
We specialize $f_0$ in (\ref{eq:u-poly})
by setting $r_i = -1$, for some $i \in \{1, \dots, n\}$,
and $r_j = 0$,
where $1 \leq j \leq n$ and $j \not= i$.
Wlog assume $r_1 = -1$.
We compute the $u-$resultant of the system,
which we call $\mathcal{R}_{1} \in \ZZ[u]$.
Its roots are the first coordinates 
of the isolated zeros of the system, viz.\ $\gamma_{1,i}$, $1 \leq i \leq D$.
Thus $\dg{\mathcal{R}_{1}} \leq D$.

The coefficients of $\mathcal{R}_{1}$ are of the form,
$\varrho_k \, \bold{c}_1^{\M_{1}} \bold{c}_2^{\M_{2}} \dots \bold{c}_n^{\M_{n}}$,
where $\varrho_k \in \ZZ$ and the interpretation of the rest of the formula is
the same as in the previous section.  
Using \cite{sombra-ajm-2004}, see also Eq.~(\ref{eq:all-def}), 
we get that
\begin{displaymath}
  | \varrho_k | \leq \prod_{i=0}^{n}{ (\#Q_i)^{\M_i}}
  = (\#Q_0)^{D} \, \prod_{i=1}^{n}{ (\#Q_i)^{\M_i}} = 2^{D} \varrho,
\end{displaymath}
since now $f_0$ is a simplex in dimension 1.
It also holds that
$|\bold{c}_1^{\M_{1}} \bold{c}_2^{\M_{2}} \dots \bold{c}_n^{\M_{n}}| 
  \leq C $.
Combining the two inequalities we deduce that 
\begin{displaymath}
  \norm{ \mathcal{R}_{1}}_{\infty} \leq 2^{D} \, \varrho \, C,
\end{displaymath}
and also $\norm{ \mathcal{R}_{1}}_{\infty} \leq \norm{\mathcal{R}_{1}}_2 \leq 2^{D} \, \sqrt{D+1} \, \varrho \, C$.

From Cauchy's bound for the roots of univariate polynomials, e.g.~\cite{Mign91},
we know that for all the roots of $\mathcal{R}_{1}$ it holds that 
$(2^{D} \, \varrho \, C)^{-1} \leq 1/\norm{\mathcal{R}_{1}}_{\infty} 
\leq |\gamma_{i,j}| \leq 
\norm{\mathcal{R}_{1}}_{\infty} \leq 2^{D} \, \varrho \, C$.
The inequality holds for all the indices $i$ and $j$.
Hence, all roots of the system in $(\CC^{*})^{n}$
are contained in a high-dimensional annulus in $\CC^{n}$,
defined as the difference of the volumes of two spheres centered at the origin, 
with radii $2^{D} \, \varrho \, C$ and $(2^{D} \, \varrho \, C)^{-1}$, resp.
This proves Eq.\ (\ref{eq:upper-lower-D}).

Now we are ready to prove the upper bound of
Eq.~(\ref{eq:dmm-D}) in Th.~\ref{th:DMM-n}.
For all $a, b \in \CC$ it holds that 
\begin{equation}
  \abs{ a - b} \leq 2 \max\Set{ \abs{a}, \abs{ b}}.
  \label{eq:a-b-ineq}
\end{equation}
Let the multiset $\overline{ \Omega} = \Setbar{ j}{ (i,j) \in \Omega}$,
where $\Abs{ \overline{ \Omega}} = \ell$, then
\begin{displaymath}
  \prod_{(i, j) \in \Omega}{ \Abs{ \gamma_i - \gamma_j}}
  \leq 2^{\ell} \,  \prod_{j \in  \overline{ \Omega}}{ \Abs{ \gamma_j}}
  \leq 2^{\ell} \, (2^{D} \, \varrho \, C)^{\ell} 
  \leq ( 2^{D+1} \, \varrho \, C)^{\ell}.
\end{displaymath}

For proving (\ref{eq:sep-D}), 
let $(i, j)$ be the pair of indices where the separation bound 
of $(\Sigma)$ is attained.
Then
\begin{displaymath}
  \sep(\Sigma) = 
  \abs{ \gamma_i - \gamma_j}  =
  \sqrt{ \sum_{k=1}^{n}{ (\gamma_{i,k} - \gamma_{j,k})^2}}
  \geq 
  \abs{ \gamma_{i, 1} - \gamma_{j, 1}}
  \geq \mathsf{sep}(\mathcal{R}_{1}),
\end{displaymath}
where $k$ is any index such that $\gamma_{i, k} \not= \gamma_{j, k}$
and $\mathsf{sep}(\mathcal{R}_{1})$ is the separation bound of $\mathcal{R}_{1}$.
An easy bound on the latter can be derived by applying 
Th.~\ref{th:DMM-1} to $\mathcal{R}_{1}$ with $\ell=1$:
$
  \mathsf{sep}( \mathcal{R}_{1}) \geq 2^{1 - {D \choose 2}} \, \norm{\mathcal{R}_{1}}_2^{-D}
  \geq 2^{(1-D)(D+2)/2} \, (D+1)^{-D/2} \, (2^{D} \varrho \, C)^{-D}
  \geq 2^{-(3D+2)(D-1)/2} \,( \sqrt{D+1} \, \varrho \, C)^{-D},
$
which completes the proof of (\ref{eq:sep-D}).  

\begin{remark}
  It is tempting to try to prove the lower bound of Th.~\ref{th:DMM-n} 
  by applying $\dmm_1$ to $\mathcal{R}_{1}$, instead of $U$, as we did in the previous
  section. This would allow us to eliminate the factor
  $B^{-(n-1)(D^2+\ell(D+1)-D)}$ from the result.
  However, if we apply $\dmm_1$ to $\mathcal{R}_{1}$, it is not obvious 
  that the requirements of Th.~\ref{th:DMM-1} hold,
  i.e.\ that the ordering of (the coordinates of) the roots is preserved.
  Moreover, the bounds on the $u-$resultant are of independent interest, 
  since the latter is used in many algorithms for system solving,
 e.g.~\cite{BPR06,EmiPhd,Rou:rur:99}.
\end{remark} 

\section{Comparisons and extensions}\label{sec:opt-ext}

One of the first multivariate separation bounds was due to Canny, later
generalized to the case when only the affine part of the variety is
0-dimensional \cite{Yap2000}.
\begin{theorem}[Gap theorem]
  \label{th:gap}
  \cite{c-crmp-87}
  Let $f_1( \bold{x}),\dots, f_n( \bold{x})$
  be polynomials of degree $d$ and coefficient magnitude $c$,
  with finitely-many common solutions when homogenized. 
  If $\gamma_j \in\CC^n$ is such a solution, then for any $k$, either
  $\gamma_{j, k} = 0 \text{ or } |\gamma_{j, k}| > \left( 3dc \right)^{-n\,d^n}$.
\end{theorem}

Let $\bitsize{f_i} =\tau$, then this becomes
$2^{(\lg{3} + \lg{d} +\tau)nd^{n}}$, which is worse than the bound in
Eq.~(\ref{eq:upper-lower-dt}), by a factor of $O(d^{n-1})$.
In \cite{by-issac-2009}, they only require that the system has a 0-dimensional
projection; $m$ is the number of polynomials and $b < n$ the
dimension of the prime component where the 0-dimensional projection
is considered.
The bound is:
\begin{displaymath}
  |\gamma_{ij}| \geq 
  ( (n+1)^2 \, e^{n+2})^{-n(n+1) d^n}
  (b^{n-b-1} \, m \, 2^{\tau})^{-(n-b)d^{n-b-1}} ,
\end{displaymath}
This is similar to ours in (\ref{eq:upper-lower-D}), and we make a comparison in the sequel. 
Moreover, Cor.\ \ref{CorMVs} does not depend on the (total) degree of the equations,
but rather on mixed volume, which is advantageous for sparse systems.

A natural question is how close are the bounds to optimum.
Let us consider the following system \cite{c-crmp-87}:
\begin{displaymath} 
    2^{\tau} x_1^2 = x_1, \, x_j  = x_{j-1}^d, \, 2 \leq j \leq n .
\end{displaymath}
The roots are $x_j = \left( 2^{-\tau}\right)^{d^{j-1}}$, for $2^{\tau} \gg 1$. 
Th.~\ref{th:DMM-n} implies $x_j \geq 2^{-d^n - n(\tau + n\lg{d} + 1)d^{n-1}}$,
which, if $\tau \gg d$, is off only by a factor of $2^{n}$ asymptotically.
The negative exponent of our bound is $\OO(n (n\lg{d} + \tau)  d^{n-1})$,
Canny's bound gives a negative exponent of 
${(\lg{3} + \lg{d} +\tau)nd^{n}} = \OO(n \tau d^{n})$.
The bound from \cite{by-issac-2009}
has negative exponent:
$n(n+1)(2\lg(n+1)+n+2)d^n$ $+ n (\lg{n} + \tau) d^{n-1}$
$= \OO(n^3 d^n + n\tau d^{n-1})$.
%


We now consider the case that $(\Sigma)$ is not 0-dimensional.
Then, the bounds of Th.~\ref{th:DMM-n} do not hold
because they are based on bounding the
infinite norm of the $u-$resultant, which is identically zero.
Specifically, the (sparse) resultant vanishes identically when 
the specialized coefficients of the polynomials are not generic enough, 
i.e.\ the variety has positive dimension,
or, simply, if the variety has a component of positive dimension 
at infinity, known as excess component.

To overcome the latter, Canny introduced the Generalized Characteristic
Polynomial (GCP) \cite{canny-gcp-1990} for dense systems.
We use its generalization, called Toric GCP (TGCP) \cite{de-cm-2001}.
We consider $(\Sigma_0)$ in~(\ref{eq:over-system}) and perturb it:
\begin{displaymath}
  (\widetilde \Sigma_0) \quad 
  \left \{      \begin{array}{lr}
        {\widetilde f_0} = f_0 = 0 , \\
        {\widetilde f_i} = f_i  + p_i = 0, & 1 \leq i \leq n ,
      \end{array}
    \right.
  \end{displaymath}
where $p_i = \sum_{\bold{a} \in \mathcal{D}_i} s^{\omega_i(\bold{a})} \bold{x}^{\bold{a}}$,
$\omega_i( \cdot)$ are (suitable) linear forms, $s$ a new parameter, and 
$\mathcal{D}_i$ is the subset of vertices in $Q_i$ corresponding to
monomials of $f_i$ on the diagonal of some sparse resultant matrix;
at worst, $\mathcal{D}_i$ contains the vertices of $Q_i$.
This perturbation does not alter the support of
the polynomials nor the mixed volume of the system. 

The TGCP is the sparse resultant of $(\widetilde \Sigma_0)$, denoted
$T\in$ $(\ZZ[\bold{c}, \bold{r}])[u,s]$, where $\bold{c}$ corresponds
to the coefficients of $f_i$ and $\bold{r}$ to the coefficients of
$f_0$. The lowest-degree nonzero coefficient of $T$, seen as
univariate polynomial in $s$, is a projection operator: 
it vanishes on the projection of any 0-dimensional component of the algebraic set
defined by $(\Sigma_0)$. We call this $T_U \in \ZZ([\bold{c}, \bold{r}])[u]$, 
and $\dg{T_U} \leq \M_0$.
The roots of $T_U$ are the isolated points of the variety
plus some points embedded in its positive-dimensional components.
It remains to bound the coefficients of $T_U$.
Repeating the construction of $U$ in Eq.~(\ref{eq:U-monom}), we get
\begin{displaymath}
  T_U = \dots + 
  \underbrace{
  {\varrho}_k \, 
  u^k \bold{r}_k^{\M_0-k} \, 
  {\widetilde {\bold{c}}_{1,k}}^{\M_{1}} \,
  {\widetilde {\bold{c}}_{2,k}}^{\M_{2}} \,
  \cdots 
  {\widetilde {\bold{c}}_{n,k}}^{\M_{n}}
  }_{t_k}
  \,
  + \dots ,
\end{displaymath}
where $\rho_k\in\ZZ$, and ${\widetilde {\bold{c}}_{i,k}}^{\M_{i}}$
is a monomial in the coefficients $c_{ij}, s$, of total degree $\M_i$.
It is an overestimation, wrt the height of $T$,
if we suppose that ${\widetilde {\bold{c}}_{i,k}}$ 
is obtained by adding $s^{\lambda}$ to each coefficient of $\bold{c}_{i,k}$,
where $\lambda = \max_{i, \bold{a}} \{ \omega_i( \bold{a})\}$.
If we expand ${\widetilde {\bold{c}}_{i,k}}^{\M_{i}}$, 
the absolute value of the coefficients of $s$ is bounded by 
${\M_i \choose \M_i/2} \norm{f_i}_{\infty}^{\M_i} \leq
2^{\M_i} \norm{f_i}_{\infty}^{\M_i} / \sqrt{\M_i}$.
If we expand the term $t_k$ of $T$, the degree of $s$ is bounded by
$\lambda \cdot \prod_{i=1}^n{\M_i}$, and the coefficients are bounded by
\begin{eqnarray*}
  \prod_{i=1}^{n}{\M_i} \cdot 
  |\varrho_k| \cdot |\bold{r}_k|^{\M_0-k} \cdot 
  \prod_{i=1}^{n}{ 2^{\M_i} \norm{f_i}_{\infty}^{\M_i} / \sqrt{\M_i}}
  = \\
  |\varrho_k| \cdot |\bold{r}_k^{\M_0-k}| \cdot 
  \prod_{i=1}^{n}{ \sqrt{ \M_i} \cdot 2^{\M_i} \cdot \norm{f_i}_{\infty}^{\M_i}}
  = |\bold{r}_k|^{\M_0-k} \, h \, A \, C ,
\end{eqnarray*}
since every factor  ${\widetilde {\bold{c}}_{i,k}}^{\M_{i}}$, 
contributes at most $\M_i$ coefficients. The bound holds for (the absolute
of) all the coefficients of $T$ if we consider it as  bivariate polynomial in $s, u$.
Recall that $|\varrho_k| \leq h$, for all $k$,
where $h$ is defined in Eq.~(\ref{eq:all-def}).
This expression also defines $A, C$.

Now $k \leq \M_0$. 
If we consider $T_U$ as a univariate polynomial
in $s$, then its coefficients are univariate polynomials in
$u$, with degree $\le \M_0$.
For the 2-norm of $T_U$, we use a summation as in the 0-dimensional case,
and get  
\begin{displaymath}
  \norm{T_U}_\infty \leq \norm{ T_U}_{2} \leq 2 \,h \,A \, C \, B^{(n-1)\M_0} .
\end{displaymath}

The previous bound is the one on $U$ multiplied by $A$. 
Thus we can provide a theorem extending Th.~\ref{th:DMM-n} to positive-dimensional systems, 
by replacing $C$ by $AC$, in Th.~\ref{th:DMM-n},
\begin{framed}
  \vspace{-5pt}
  \begin{theorem}[$\texttt{DMM}_n$ with excess components]
    \label{th:DMM-n-pos}
    Consider the polynomial system $(\Sigma)$ in
    (\ref{eq:system}), 
    which is not necessarily 0-dimensional,
    and where it holds that  $f_i \in \ZZ[ \bvec{x}]$,
    $\dg{f_i} \leq d$,
    and $\bitsize{f_i} \leq \tau.$
    Let $D$ be the number of the isolated points of the solution set in
    $(\CC^*)^n$, which are $0 < | \gamma_1 | \leq |\gamma_2| \leq
    \dots \leq |\gamma_D|$.  Let $\Omega$ be any set of $\ell$ couples
    of indices $(i, j)$ such that $1 \leq i < j \leq D$,
    and  $\gamma_{j,k}$ stands for the $k$-th coordinate of  $\gamma_j$.
    Then the following holds
      \begin{equation*}
        (2^{\M_0+1} \, \varrho \, C \, A)^{\ell} 
            \geq
            \prod_{(i, j) \in \Omega}{ \Abs{ \gamma_i - \gamma_j }}
           \geq
          2^{-\ell - (\M_0-1)(\M_0+2)/2} \, (h \,C \, A)^{1-\M_0-\ell} \, B^{(1-n)(\M_0^2 + \M_0(\ell-1) + \ell)},
        \label{eq:dmm-D-pos}
      \end{equation*}
    \begin{equation}
      (2^{\M_0} \, \varrho \, C \, A)^{-1}
      \leq  |\gamma_{j,k} |  \leq 
      2^{\M_0} \, \varrho \, C \, A,  \label{eq:upper-lower-D-pos}
    \end{equation}  

    \begin{equation}
      \sep( \Sigma) 
      \geq 
      2^{-(3\M_0+2)(\M_0-1)/2} \,( \sqrt{\M_0+1} \, \varrho \, C \, A)^{-\M_0}, \label{eq:sep-D-pos}
    \end{equation}
    We also have the following, less accurate bounds:
    \begin{equation}
      \prod_{(i, j) \in \Omega}{ \Abs{ \gamma_i - \gamma_j }}
      \geq 
      2^{-(n^2-n)d^n\lg{\sqrt{d}} -(3+4\lg{n}+4n\lg{d})d^{2n}}
      \, \cdot \,2^{-2n(2+n\lg{d}+\tau)d^{2n-1}},
      \label{eq:dmm-dt-pos}
    \end{equation}    
    \begin{equation}
      2^{(n^2-n)\lg{\sqrt{d}} -d^n - n(\tau + n\lg{d} + 2)d^{n-1}}
      \leq |\gamma_{j,k} | \leq 
       2^{(n^2-n)\lg{\sqrt{d}} +d^n + n(\tau + n\lg{d} + 2)d^{n-1}}, 
      \label{eq:upper-lower-dt-pos}
    \end{equation}
    \begin{equation}
      \mathsf{sep}( \Sigma) \geq 
      2^{-(n^2-n)d^n\lg{\sqrt{d}} -2d^{2n} -n(2n\lg{d} + \tau+1)d^{2n-1}}.
      \label{eq:sep-dt-pos}
    \end{equation}
  \end{theorem}
\end{framed}

\section{Applications}
\label{sec:applications}

We illustrate the bounds of Th.~\ref{th:DMM-n} in two applications.
The first concerns matrix eigenvalues and eigenvectors,
and is a standard illustration of the superiority of mixed volumes
against B\'ezout's bound.
The second is lower bounds of positive multivariate
polynomials, inspired by \cite{bsr-arxix-2009}.

{\bf Eigenvalues and eigenvectors.}
Consider an $n \times n$ integer matrix $A$, with elements $\leq 2^{\tau}$. 
We are interested in its eigenvalues $\lambda$,
and its eigenvectors $\bvec{v} = (v_1, \dots, v_n)^{\top}$.
This is equivalent to solving
$f_j = \sum_{j=1}^{n}{a_{i,j} v_j} - \lambda v_i$,
$1 \leq i \leq n$, $1 \leq j \leq n$, and $f_{n+1} = \sum_{i=1}^{n}{v_i^2} - 1$.
We have $\norm{f_j}_{\infty} \leq 2^{\tau}$,
$\norm{f_{n+1}}_{\infty} \leq 2$.
The B\'ezout bound is $2^{n+1}$, whereas
the actual number of (complex) solutions is $2n$, 
which equals the mixed volume, e.g.~\cite{EmiPhd}.

Canny's Gap theorem \cite{c-crmp-87} implies
$|z| > (6 \cdot 2^{\tau})^{-(n+1)2^n}$,
for any eigenvalue or eigenvector element $z\ne 0$.
Thus, 
we need $\OO( n \, \tau \, 2^{n})$ bits. 
We get the same exponential behavior in $n$ if we apply
\cite{Yap2000} or \cite{by-issac-2009}.

It is reasonable to assume that the system is 0-dimensional and
apply (\ref{eq:upper-lower-D}) of Th.~\ref{th:DMM-n}.
It holds that $\M_j = 2n$, $\M_{n+1} = n$,
$(\#Q_{n+1}) \leq 2^{n+2}$, and $(\#Q_i) \leq 2^{n+2}$
where $1 \leq j \leq n$, and 
$
  C = 
  \norm{ f_{n+1}}_{\infty}^{\M_{n+1}} \, \prod_{j=1}^{n}{ \norm{f_j}_{\infty}^{\M_{j}}} 
  \leq 2^{\tau \sum_{j=1}^{n}{\M_j}} \ 2^{n} = 2^{2 n^2 \tau + n},
$
$
\varrho \leq \prod_{i=1}^{n+1}{(\#Q_i)^{\M_i}} 
\leq (\#Q_{n+1})^{\M_{n+1}} \, \prod_{i=1}^{n}{(\#Q_i)^{\M_i}}$;
hence 
$\varrho \leq
(2^{n+2})^{n} \, \prod_{i=1}^{n}(2^{n+2})^{2n} \leq 2^{2n^3 + 5n^2+ 2n}.
$

The solutions lie 
in $\CC^{n+1}$.
The lower bound of Th.~\ref{th:DMM-n} yields
\begin{displaymath}
|z| > 2^{-2n^3 - 5n^2 - 5 - 2n^2\tau}, 
\end{displaymath}
where $z$ is an eigenvalues or element of eigenvector.
This is exponentially better than the previous bounds.  
Eq.~(\ref{eq:sep-D}) from Th.~\ref{th:DMM-n} bounds the system's separation bound:
$-\lg( \sep( \Sigma)) \leq 
  4 n^3 \tau + n\lg{n} + 4\,{n}^{4}+10\,{n}^{3}+12\,{n}^{2}+n-1 = \OO( n^4 + n^3 \tau)$. 
This is polynomial in the size of the input, and hence we obtain a new proof
of Bareiss' result \cite{bareiss-moc-1968}, that computing the
eigenvalues and eigenvectors of an integer matrix is a polynomial problem.

{\bf Positive multivariate polynomials.}
We consider the following problem, studied in \cite{bsr-arxix-2009}.
Let $P \in \ZZ[x_1, \dots, x_n]$ be a multivariate polynomial of degree $d$
which on the $n$-dimensional simplex
takes only positive values.
We are interested in computing a bound on its {\em minimum value}, $m$.
We may assume that the minimum is attained inside the simplex; if not,
apply a transformation which slightly changes the bitsize of $P$ \cite{bsr-arxix-2009}.
Let $\tau$ bound the bitsize of the coefficients of $P$.
We wish to find compute a lower bound on $m$,
 greater than zero, depending on $n, d, \tau$.
Equivalently, we have a system with unknowns $m, x_i$:

\begin{equation}
  \left\{
  \begin{array}{l}
    \frac{\partial P}{\partial x_1}(x_1, \dots, x_n) 
    = \cdots = \frac{\partial P}{\partial x_n}(x_1, \dots, x_n) = 0, \\
    P(x_1, \dots, x_n) = m.
  \end{array}
  \right.
  \label{eq:pos-poly-system}
\end{equation}
We use Th.~\ref{th:DMM-n-pos}, since there is no guarantee that the system is 
0-dimensional. However, Th.~\ref{th:DMM-n-pos} provides bounds for the
isolated points of the variety. Since the  minimum could be attained on a non-zero
dimensional component, we should argue that the bounds take care
of this case.
We consider all the irreducible components of the variety defined by
\eqref{eq:pos-poly-system}.
Each of them contains a point for which the bounds of
Th.~\ref{th:DMM-n-pos} apply. 
Such a point is the limit of a solution
of the perturbed system depending on the parameter $s$ when
$s\rightarrow 0$. Moreover, it is a zero of the first non-zero
coefficient $T_{U}$, seen as a polynomial in $s$
\cite{canny-gcp-1990,de-cm-2001};
Th.~\ref{th:DMM-n-pos} bounds these zeros.
Now, on each of these components, the value of $m$ is constant,
since the gradient of $P$ is~0, and so the bounds apply for it as well.

Let $P_i =\frac{\partial P}{\partial x_i}$ and $P_{n+1} = P -m$.
It holds that
$\deg(P_{n+1}) = d$, $\deg( P_i) \leq d-1$,
$\norm{P_{n+1}}_{\infty} \leq 2^{\tau}$, 
$\norm{P_i}_{\infty} \leq d \norm{f_{n+1}}_{\infty} \leq d\,2^{\tau}$,
$\M_{n+1} \leq (d-1)^n$,
$\M_{i} \leq d (d-1)^{n-1}$,
and $D \leq \M_0 \leq d (d-1)^n$.
Using (\ref{eq:upper-lower-D-pos}) we deduce $1/m \leq 2^{D} \, \varrho\,C\,A$.
It remains to bound the various quantities involved, defined in
(\ref{eq:all-def}):
{
\begin{displaymath}
  \begin{aligned}
    C & \leq \prod_{i=1}^{n+1}{ \norm{P_i}_{\infty}^{\M_i}}
    = \norm{P_{n+1}}_{\infty}^{\M_{n+1}} \prod_{i=1}^{n+1}{\norm{P_i}_{\infty}^{\M_i}} \\
    & \leq \Paren{2^{\tau}}^{(d-1)^n} \prod_{i=1}^{n}{ \Paren{d\,2^{\tau}}^{d(d-1)^{n-1}}} 
    \leq 2^{\tau (d-1)^n}  (d\,2^{\tau})^{ nd(d-1)^{n-1}} \\  
    & \leq  2^{ (n+1) \tau  d (d-1)^{n-1}  + n d (d-1)^{n-1}\lg{d}  },
\end{aligned}
\end{displaymath}
}
{
\begin{displaymath}
  \begin{aligned}
    A & = \prod_{i=1}^{n+1} \sqrt{\M_i} \, 2^{\M_i} =
    \sqrt{\M_{n+1}} \cdot 2^{\M_{n+1}} \cdot \prod_{i=1}^{n} \sqrt{\M_i} \cdot 2^{\M_i} \\
    & \leq (d-1)^{n/2} \cdot 2^{(d-1)^n} \cdot d^{n/2} (d-1)^{n(n-1)/2} \cdot 2^{nd(d-1)^{n-1}} \\
    & \leq 2^{(n+1)d(d-1)^{n-1} + (n^2+n)\lg{\sqrt{d}}}.
  \end{aligned}
\end{displaymath}
}
Moreover,
$(\#Q_{n+1}) \leq 2d^{n+1}$, $(\#Q_i) \leq 2(d-1)^{n+1}$, and so 
{
\begin{displaymath}
  \begin{aligned}
    \varrho & = \prod_{i=1}^{n+1}{(\# Q_i)^{M_i}} = {(\# Q_{n+1})^{M_{n+1}}} \prod_{i=1}^{n}{(\# Q_i)^{M_i}}  \\
    & \leq (2 d^{n+1})^{(d-1)^n} \cdot \prod_{i=1}^{n}{(2d^n)^{d(d-1)^{n-1}}}
     \leq 2 ^{(n+1)(1+(n+1)\lg{d})d (d-1)^{n-1}}
  \end{aligned}
\end{displaymath}
}
We apply (\ref{eq:upper-lower-dt}) using the previous
inequalities, and get 

\begin{displaymath}
  \frac{1}{m} \leq 
  2^{(n^2+ n)\lg{\sqrt{d}} +  (1+2n + d + (n^2+3n+1)\lg{d})d(d-1)^{n-1}} 
  \cdot  \, 2^{(n+1)\tau d(d-1)^{n-1}} .
\end{displaymath}
To assure that the minimum is attained inside the simplex, we apply a
transformation that preserves the degree, but the bitsize of the polynomial is
now bounded by $\tau +1 + d \lg(n)$. Replacing this in the previous inequality,
we get $ \frac{1}{m} \leq \frac{1}{m_{\dmm p}}$, where
\begin{equation}
  \frac{1}{m_{\dmm p}}=
  2^{(n^2+ n)\lg{\sqrt{d}} + (2+ 3n + d + (n^2+3n+1)\lg{d}}  
  \cdot \, 2^{(n+1)d\lg{n})d(d-1)^{n-1}} \cdot \, 2^{ (n+1)\tau d(d-1)^{n-1}} .
  \label{eq:pos-poly-bound}
\end{equation}

If we know that the system is zero dimensional then we could use
Th.~\ref{th:DMM-n}. Of course this is not always the case, hence we
state the following bound, using (\ref{eq:upper-lower-D}), just as a reference.
\begin{equation}
  \frac{1}{m} \leq \frac{1}{m_{\dmm}}=
  2^{((n+1)\tau + n+d+(n^2+3n+1) \lg{d} )d(d-1)^{n-1}}.
  \label{eq:pos-poly-bound-0}  
\end{equation}

Let us compare the $m_{\dmm_{p}}$ with other bounds that appear in the bibliography.
In \cite[Sec.~2, Rem.~2.17]{bsr-arxix-2009}, the following estimation was
computed, 
\begin{equation}
  \begin{aligned}
    \frac{1}{m_{\texttt{BLR}}} & =
    2^{2^{n+3} n d^{n+1} (\tau + 8 nd)} n^{2^{n+5}d^{n+2}n}d^{2^{n+5}d^{n+1}n^{2}} \\
    & = 2^{ 2^{n+3} n \tau d^{n+1}+
      2^{n+5} n d^{n+1} (2 nd + d\lg{n} + n\lg{d}) } .
    \label{eq:blp-pos-bound}
  \end{aligned}
\end{equation}
which also holds with no assumption, but it is looser than
$m_{\dmm_{p}}$.  

In \cite{by-issac-2009} the authors derive a bound for the minimum of
the absolute value of a polynomial, 
$\frac{1}{m} \leq \frac{1}{m_{\texttt{BY}}}$,
i.e.
\begin{equation}
  \frac{1}{m_{\texttt{BY}}}  =
  ((n+2)^2 e^{n+3})^{(n+1)(n+2)d^{n+1}} \, (n^n (n+1)\,d \,2^{\tau})^{(n+1)d^n}.
  \label{eq:by-eval-bound}
\end{equation}
The authors use the terminology {\em evaluation bound} for their
bound. It holds when there is a 0-dimensional projection;
they prove that this is always the case for~(\ref{eq:pos-poly-system}).

In \cite{jp-arxiv-2009} the following bound was computed:
\begin{equation}
  \frac{1}{m} \leq 
  \frac{1}{m_{\texttt{JP}}}  = 2^{(\tau+1)  d^{n+1}} d^{(n+1) d^{n+1}} ,
  \label{eq:pg-pos-bound}
\end{equation}
which has no restriction on the corresponding polynomial system.
It is comparable to $m_{\dmm_{p}}$ in general, but strictly looser 
when $d > n$. 

\begin{example}
  Let us compute a lower bound on the value of
  $f = (x +2y -3)^d + (x+ 2y -4)^d$,
  where $d \in \set{2,8,32}$. The polynomial is positive as it is a
  sum of squares. 
  Consider the ideal $I=(f-z, f_x, f_y) \subset \ZZ[x,y,z]$.
  If $(\zeta_1, \zeta_2, \zeta_3)$ belongs to the zero set of $I_f$,
  then $|\zeta_3| \geq 2^{-b}$, $b>0$.
  In Tab.~\ref{tab:pos-poly}
  we present the estimations of $\lg{b}$ by the previous bounds.
  The true value is $b = 0$.
  When the degree is comparable to the number of variables ($d=2$), then
  our bound and $m_{\texttt{JP}}$ are comparable.
  When $d > n$, e.g.~$d=4$ and $d=32$, 
  then $m_{\dmm p}$  is better than $m_{\texttt{JP}}$ by an order of
  magnitude.
\end{example}

\section{Subdivision algorithms}
\label{sec:subdiv}

We use our results to bound the number of steps that any subdivision algorithm
performs to isolate the real roots of a well-defined polynomial system.
Then, we bound the complexity of Milne's algorithm in 2d.
Our analysis can easily be extended to $\RR^n$, however it is not clear what
is the exact bit complexity of the elimination steps needed.


We use $\dmm_n$, Th.~\ref{th:DMM-n}, and Eq.~(\ref{eq:dmm-D}) \& (\ref{eq:dmm-dt}),
to bound the number of steps of a subdivision algorithm
to isolate the real roots of a well-defined polynomial system as in~(\ref{eq:polys}).
We assume the existence of an oracle that counts the number of real roots
of the system inside a box in $\QQ^{n}$.
Our aim is to compute the number of calls to the oracle in order to
compute isolating (hyper-)boxes for all real roots.
Realizations of such oracles are in \cite{Miln92,PeRoSz93,p-crz-91}, see also \cite{BPR06}.

Suppose all roots of the system lie in a hypercube of side $C$, see Th.~\ref{th:DMM-n}. 
At step $h$ of the algorithm, the oracle counts the number of roots 
in hypercubes of side $C/2^{h}$.  
We consider the whole subdivision algorithm as a $2^{n}-$ary tree $T$,
where at each node we associate a hypercube, and to the root of the
tree we associate the initial hypercube. 
Let $\#(T)$ denote the number of nodes. We will prune some leaves of $T$ to
obtain tree $T'$ where it is easier to count its nodes.

We proceed as follows.
If $v$ is a leaf and has a sibling that it is not a leaf, then we prune $v$.
If $u_1, \dots, u_k$, for some positive integer $k$,
are leaves and siblings, such that they have no
sibling that is not a leave, then we prune all of them except one that
possess a hypercube that contains a real root. 
Notice that there is always at least one such node in
$u_1, \dots, u_k$, because otherwise, the subdivision process in this
path would have stopped one level before. 
If there exists more than one such node in $u_1, \dots, u_k$, then we keep
arbitrarily one of them. 
It holds that $\#(T) \leq 2^{n} \#(T')$, and we will count the nodes in $T'$.

Each leaf of the tree contains contains a hypercube that isolates a
real root of the system, and if there are at most $R$ real roots,
this also bounds the number of the leaves of $T'$.
The hypercubes that correspond to the leaves of the tree have diagonals 
that are at least $\Delta_j = |\gamma_j - \gamma_{c_j}|$,
where $\gamma_{c_j}$ is the root closest to $\gamma_j$.
The length of their edges is at least 
$| \gamma_{j, i} - \gamma_{{c_j}, i} |$, where $1 \leq i \leq n$.
It holds that
$
  \Delta_j = |\gamma_j - \gamma_{c_j}| \geq 
  | \gamma_{j, i} - \gamma_{{c_j}, i} |,
$
for any index $i$.
The number of nodes from a leaf to the root of the tree is
$\Ceil{ \log{ \frac{ C}{\Delta_j}}}$.
Hence the number of nodes in $\#( T')$ is
\begin{equation}
  \label{eq:T-complex}
  \#( T') = \sum_{j=1}^{R}{ \Ceil{ \log{ \frac{ C}{\Delta_j}}}} \leq
  R + R \lg{C} - \lg{ \prod_{j=1}^{R}{\Delta_j}}.
\end{equation}

To bound the various quantities that appear, we will rely on
Eq.~(\ref{eq:all-def}) and Th.~\ref{th:DMM-n}.
If the total degree of the polynomials is bounded by $d$, 
and $\norm{f_i}_{\infty} \leq 2^{\tau}$, 
then $\lg{ C} \leq n \, \tau \, d^{n-1}$.
To bound $\prod_{j=1}^{R}{\Delta_j}$ 
we use Eq.~(\ref{eq:dmm-D}) of Th.~\ref{th:DMM-n} with $\ell = R$.
The hypotheses of the theorem, concerning the indices of the roots,
are not fulfilled when symmetric products occur.
In this case, we factorize quantity as
$\prod_{i=1}^{R}{\Delta_i} = \prod_{i=1}^{R_1}{\Delta_i} \prod_{i=1}^{R_2}{\Delta_i}$,
where $R_1 + R_2 = R$ and the factors are such that no symmetric
products occur. Then
\begin{eqnarray*}
  &\prod_{i=1}^{R}{\Delta_i} = \prod_{i=1}^{R_1}{\Delta_i} \prod_{R=1}^{R_2}{\Delta_i}
  \geq \\
  &2^{-R - (D-1)(D+2)} \, (h\, C)^{2 - 2D - R} \, B^{-(n-1)(2D^2 + D(R+2) + R)}.
\end{eqnarray*}
If we take into account that $R \leq D \leq d^n$, then 
\begin{displaymath}
  \begin{array}{ll}
    -\log{ \prod_{i=1}^{R}{\Delta_i}} &\leq 
    2 D^2 + 3D \lg{C} + 3D \lg{h} + 5 n \,D^2 \lg{B} \\
    &\leq
    8(\lg{n}+n\lg{d})d^{2n} + 3n(n\lg{d}+\tau)d^{2n-1},
  \end{array}
\end{displaymath}
and for the total number of nodes of $T'$ we have
\begin{displaymath}
  \hspace{-0.3cm}
  \begin{array}{rl}
    (\#T')  \leq &
    R + R \lg{C} - \lg{ \prod_{j=1}^{R}{\Delta_j}}  \leq D + D \lg{C}  - \lg{ \prod_{j=1}^{R}{\Delta_j}} \\
    \leq  &
    2 d^n (n \tau d^{n-1})  +8(\lg{n}+n\lg{d})d^{2n} + {3n(n\lg{d}+\tau)d^{2n-1}} \\
     = &   
    \sO( n (n + d + \tau) d^{2n-1}),
  \end{array}
\end{displaymath}
and hence $(\#T) = \sO( 2^{n} \, n (n + d + \tau) \, d^{2n-1})$.
\begin{theorem}
  \label{th:subdiv-steps}
  Consider the polynomial system formed by the polynomials in (\ref{eq:polys}).
  The number of steps that a subdivision algorithm performs in order to compute isolating
  boxes for all the real roots of the system is
  $\sO( 2^{n} \, D \,(D  + \lg{C}))$
  or
  $\sO( 2^{n} \,( d + \tau) \, d^{2n-1})$.
\end{theorem}

\begin{remark}
  If we specialize $n = 1$ in the previous theorem, then we deduce 
  that the number of steps of subdivisions algorithms for real root isolation of
  univariate integer, not necessarily square-free, polynomials 
  is $\OO( d^2 \lg{d} + d \tau)$.
  The optimal bound is $\OO( d^2 + d \tau)$ \cite{Dav:TR:85}.
\end{remark}


We now bound the complexity of Milne's algorithm \cite{Miln92}
for isolating all real roots of a bivariate polynomial system.
Milne's, so-called, {\em volume function} realizes the required oracle,
see \cite{gvt-msturm-1997,z-msc-2009} for experimental results.  
By $\SR(f, g)$ we denote the signed polynomial remainder sequence of $f, g$.

\begin{proposition}
  \label{prop:mSR-comp} 
  \cite{Reischert:subresultant:97,det-jsc-2009}
  We compute $\SRQ( f, g)$, any polynomial in $\SR( f, g)$, 
  and $\res( f, g)$ wrt $x$ in $\sOB( q (p+q)^{k+1} d \tau)$.
  The degree of $\SR( f, g)$ in $y_1, \dots, y_k$ is $\OO( d(p+q))$ and the bitsize is
  $\OO( (p+q)\tau)$.
  We can evaluate $\SR( f, g )$ at $x = \rat{ a}$,
  where $\rat{a} \in \QQ \cup \{ \infty \}$ and $\bitsize{ \rat{ a}} = \sigma$,
  in $\sOB( q (p+q)^{k+1} d \max\{ \tau, \sigma\})$.
\end{proposition} 

Let $f,g \in \ZZ[x, y]$ with total degrees bounded by $d$ and bitsize bounded by
$\tau$.
We are interested in isolating the real roots of the polynomial system 
$f(x, y) = g(x, y) = 0$,
which we assume to be 0-dimensional.
We introduce new parameters $u$, $a$, $b$ and we eliminate $a$, $b$ from the
polynomials $ \set{ f(a, b),  g(a, b), V = u + (x-a)(y-b)}$,
where $V$ is the volume function.
After elimination, we obtain a polynomial $h \in (\ZZ[x, y])[u]$.
We compute the Sturm sequence of $h$ and its derivative w.r.t. $u$,
$h_u$, and we evaluate the sequence over $u=0$. We obtain a sequence
of bivariate polynomials in $x$, $y$.
Now consider a box in the plane.
We evaluate the sequence on each vertex of the box, and we count the number of
sign variations.
The number of real roots inside the box is $\tfrac{1}{4}$ the sum
of the sign variations \cite{Miln92}.  

We perform the elimination using iterated resultants.
Using Prop.~\ref{prop:mSR-comp} we compute 
$h_1 = \res_a( f(a, b), V(u, x, y, a, b)) \in \ZZ[u, x, y, a, b]$
in $\sOB( d^7 \tau)$.
The total degree of $h_1$ is $\OO( d^2)$ and $\bitsize{ h_1} = \sO( d \tau)$.
Similarly, we obtain polynomial
$h_2 = \res_a( g(a, b), V(u, x, y, a, b)) \in \ZZ[u, x, y, a, b]$.
Finally, $h = \res_b( h_1, h_2) \in \ZZ[x, y, u]$
is computed in $\sOB( d^{12} \tau)$.
The degree of $h$ in $u$ is $\OO( d^2)$ since the resultant of $h_1,
h_2$ has the factor $u^{ \dg{ f(x, 0)} \dg{ g( x, 0)}} = u^{d^2}$.
The degree of $h$ in $x, y$ is $\sOB( d^4)$ and $\bitsize{ h} = \sO( d^3 \tau)$.

We compute the signed polynomial remainder sequence of $h, h_u$ and evaluate
it at 0. This costs $\sOB( d^{15} \tau)$.
The evaluated sequence contains $\OO( d^2)$ polynomials in $\ZZ[x, y]$
of degrees $\sOB( d^6)$ and bitsize $\sOB( d^5 \tau)$.  
Each polynomial in the sequence is evaluated over a rational number of
bitsize $\sigma$ in $\sOB( d^{17}( \tau + d \sigma))$,
and thus all of them in $\sOB( d^{19}( \tau + d \sigma))$.

In the worst case, $\sigma$ equals the bitsize of the separation
bound, i.e.\ $\sO(d^3 \tau)$. 
Hence, the evaluation of the sequence costs $\sOB(d^{23} \tau)$.
Th.~\ref{th:subdiv-steps} indicates that we need to perform this
evaluation $\OO( d^4 \lg{d} + d^3 \tau)$ times.

\begin{theorem}
  Let $f,g \in \ZZ[x, y]$ with total degrees bounded by $d$ and bitsize bounded
  by $\tau$.  Using the algorithm of Milne \cite{Miln92}, we can isolate the
  real roots of the system $f = g = 0$ in $\sOB( d^{27} \tau + d^{26} \tau^2)$.
\end{theorem}

Bounds on mutli-point evaluation of multivariate polynomials
\cite{NusZie-esa-2004m} 
could save at least two factors in the previous theorem.

{\small {\bf Acknowledgment.}
E.T. thanks M. Sombra for finding a missing factor in the original
manuscript, and brought to our attention \cite{sombra-ajm-2004}.
IZE and BM are partially supported by 
Marie-Curie Network ``SAGA'', FP7 contract PITN-GA-2008-214584.
ET is partially supported by an individual postdoctoral grant from the
Danish Agency for Science, Technology and Innovation.
}

{
  \begin{table}[t]
    \centering
    \begin{tabular}{|c|c||r|r|r|}
      \hline
      \multicolumn{2}{|c||}{bound} 
      & $(d, \tau)=(2, 5)$ & $(8,20)$ & $(32,85)$ \\ \hline
      \cite{bsr-arxix-2009}, Eq.~(\ref{eq:blp-pos-bound})  & $|\lg({m_{\texttt{BLR}})|}$
      & $27\,136$  & $6\,684\,672$  & $1\,604\,321\,280$ \\ \hline
      \cite{by-issac-2009}, Eq.~(\ref{eq:by-eval-bound})  & $|\lg({m_{\texttt{BY}})|}$
      &  $1\,192$  &     $74\,000$  & $4\,696\,811$ \\ \hline
      \cite{jp-arxiv-2009}, Eq.~(\ref{eq:pg-pos-bound})   &  $|\lg({m_{\texttt{JP}})|}$ 
      &      $72$  &     $15\,360$  & $3\,309\,568$ \\ \hline
      Eq.(\ref{eq:pos-poly-bound})  &  $|\lg({m_{\texttt{DMM}_{p}})|}$
      &      $87$  &      $7\,457$   & $442\,447$ \\\hline \hline
      Eq.(\ref{eq:pos-poly-bound-0})  &  $|\lg({m_{\texttt{DMM}})|}$
      &      $54$  &      $5\,201$   & $324\,506$ \\\hline
    \end{tabular} 
    \caption{Comparison of (the bitsize of) various bounds on the minimum value of the
      polynomial $f = (x +2y -3)^d + (x+ 2y -4)^d$,
      for $d \in \set{2,8,32}$ and $\tau \in \set{8,20,85}$, resp.
      The bounds hold for all polynomials with same characteristics.}
    \label{tab:pos-poly}
\vspace{-0.3cm}
  \end{table}

\bibliographystyle{plain}
\bibliography{bibDMM} 

\begin{thebibliography}{10}

\bibitem{bareiss-moc-1968}
E.H. Bareiss.
\newblock {Sylvester's identity and multistep integer-preserving Gaussian
  elimination}.
\newblock {\em Math. of Comput.}, 22(103):565--578, 1968.

\bibitem{bsr-arxix-2009}
S.~Basu, R.~Leroy, and M-F. Roy.
\newblock A bound on the minimum of the real positive polynomial over the
  standard simplex.
\newblock Technical Report arXiv:0902.3304v1, arXiv, Feb 2009.

\bibitem{BPR06}
S.~Basu, R.~Pollack, and {M-F.} Roy.
\newblock {\em Algorithms in Real Algebraic Geometry}, volume~10 of {\em
  Algorithms \& Comput. in Math.}
\newblock Springer-Verlag, 2nd edition, 2006.

\bibitem{blichfeldt-tams-1914}
H.~F. Blichfeldt.
\newblock A new principle in the geometry of numbers, with some applications.
\newblock {\em Trans. AMS}, 15(3):227--235, 1914.

\bibitem{by-issac-2009}
W.~D. Brownawell and C.~K. Yap.
\newblock Lower bounds for zero-dimensional projections.
\newblock In {\em Proc.\ ISSAC}, KIAS, Seoul, Korea, 2009.

\bibitem{bcgy-issac-2008}
M.~Burr, S.W. Choi, B.~Galehouse, and C.~K. Yap.
\newblock Complete subdivision algorithms, {II}: Isotopic meshing of singular
  algebraic curves.
\newblock In {\em Proc.\ ISSAC}, pages 87--94, Hagenberg, Austria, 2008.

\bibitem{c-crmp-87}
J.~Canny.
\newblock {\em The Complexity of Robot Motion Planning}.
\newblock ACM Doctoral Dissertation Award Series. MIT Press, 1987.

\bibitem{canny-gcp-1990}
J.~Canny.
\newblock {Generalised characteristic polynomials}.
\newblock {\em J.~Symbolic Computation}, 9(3):241--250, 1990.

\bibitem{CLO2}
D.~Cox, J.~Little, and D.~O'Shea.
\newblock {\em Using Algebraic Geometry}.
\newblock Number 185 in GTM. Springer, New York, 2nd edition, 2005.

\bibitem{de-cm-2001}
C.~D'Andrea and I.Z. Emiris.
\newblock Computing sparse projection operators.
\newblock {\em Contemporary Mathematics}, 286:121--140, 2001.

\bibitem{Dav:TR:85}
J.~H. Davenport.
\newblock Cylindrical algebraic decomposition.
\newblock Technical Report 88--10, School of Math. Sciences, Univ. Bath,
  http://www.bath.ac.uk/masjhd/, 1988.

\bibitem{det-jsc-2009}
D.~I. Diochnos, I.~Z. Emiris, and E.~P. Tsigaridas.
\newblock On the asymptotic and practical complexity of solving bivariate
  systems over the reals.
\newblock {\em J. Symb. Comput.}, 44(7):818--835, 2009.

\bibitem{ESY:descartes}
A.~Eigenwillig, V.~Sharma, and C.~K. Yap.
\newblock {Almost tight recursion tree bounds for the Descartes method}.
\newblock In {\em Proc.\ ISSAC}, pages 71--78, New York, USA, 2006.

\bibitem{EmiPhd}
{I.~Z.} Emiris.
\newblock {\em Sparse Elimination and Applications in Kinematics}.
\newblock PhD thesis, Computer Science Division, Univ.\ of California at
  Berkeley, December 1994.

\bibitem{gvt-msturm-1997}
L.~Gonz\'alez-Vega and G.~Trujillo.
\newblock Multivariate {Sturm-Habicht} sequences: Real root counting on
  n-rectangles and triangles.
\newblock {\em Real Algebraic and Analytic Geometry (Segovia, 1995), Rev. Mat.
  Univ. Complut. Madrid}, 10:119--130, 1997.

\bibitem{jp-arxiv-2009}
G.~Jeronimo and D.~Perrucci.
\newblock On the minimum of a positive polynomial over the standard simplex.
\newblock {\em CoRR}, abs/0906.4377, 2009.

\bibitem{Johnson-phd-91}
J.~R. Johnson.
\newblock {\em {Algorithms for Polynomial Real Root Isolation}}.
\newblock PhD thesis, The Ohio State Univ., 1991.

\bibitem{kps-djm-2001}
T.~Krick, L.M. Pardo, and M.~Sombra.
\newblock {Sharp estimates for the arithmetic Nullstellensatz}.
\newblock {\em Duke Mathematical Journal}, 109(3):521--598, 2001.

\bibitem{MANTZAFLARIS:2009:INRIA-00387399:1}
{A}. {M}antzaflaris, {B}. {M}ourrain, and {E}.{P}. {T}sigaridas.
\newblock {C}ontinued fraction expansion of real roots of polynomial systems.
\newblock In {\em Proc.\ {S}ymbolic-{N}umeric {C}omput.}, pages 85--94,
  {K}yoto, 2009.

\bibitem{Mign91}
M.~Mignotte.
\newblock {\em Mathematics for computer algebra}.
\newblock Springer-Verlag, New York, 1991.

\bibitem{Mignotte:AAECC:95}
M.~Mignotte.
\newblock {On the Distance Between the Roots of a Polynomial.}
\newblock {\em Appl. Algebra Eng. Commun. Comput.}, 6(6):327--332, 1995.

\bibitem{Miln92}
P.~S. Milne.
\newblock On the solution of a set of polynomial equations.
\newblock In B.~Donald, D.~Kapur, and J.~Mundy, editors, {\em Symbolic \&
  Numerical Computation for AI}, pages 89--102. 1992.

\bibitem{NusZie-esa-2004m}
M.~N{\"u}sken and M.~Ziegler.
\newblock Fast multipoint evaluation of bivariate polynomials.
\newblock In S.~Albers and T.~Radzik, editors, {\em ESA}, volume 3221 of {\em
  Lecture Notes in Computer Science}, pages 544--555. Springer, 2004.

\bibitem{p-crz-91}
P.~Pedersen.
\newblock {\em Counting real zeros}.
\newblock PhD thesis, NY Univ., 1991.

\bibitem{PeRoSz93}
P.~Pedersen, M-F.\ Roy, and A.\ Szpirglas.
\newblock Counting real zeros in the multivariate case.
\newblock In F.\ Eyssette and A.\ Galligo, editors, {\em Computational
  Algebraic Geometry}, volume 109 of {\em Progress in Mathematics}, pages
  203--224. Birkh\"{a}user, Boston, 1993.

\bibitem{Reischert:subresultant:97}
D.~Reischert.
\newblock Asymptotically fast computation of subresultants.
\newblock In {\em Proc.\ ISSAC}, pages 233--240, 1997.

\bibitem{Rou:rur:99}
F.~Rouillier.
\newblock Solving zero-dimensional systems through the rational univariate
  representation.
\newblock {\em J. of Appl. Algebra in Engin., Comm. and Computing},
  9(5):433--461, 1999.

\bibitem{sombra-ajm-2004}
M.~Sombra.
\newblock The height of the mixed sparse resultant.
\newblock {\em Amer. J. Math.}, 126:1253--1260, 2004.

\bibitem{te-tcs-2008}
Elias~P. Tsigaridas and Ioannis~Z. Emiris.
\newblock {On the complexity of real root isolation using Continued Fractions}.
\newblock {\em Theor.\ Comput.\ Sci.}, 392:158--173, 2008.

\bibitem{yakoubsohn-joc-2005}
J-C. Yakoubsohn.
\newblock Numerical analysis of a bisection-exclusion method to find zeros of
  univariate analytic functions.
\newblock {\em J. Complexity}, 21(5):652--690, 2005.

\bibitem{Yap2000}
C.~K. Yap.
\newblock {\em Fundamental Problems of Algorithmic Algebra}.
\newblock Oxford University Press, New York, 2000.

\bibitem{z-msc-2009}
Z.~Zafeirakopoulos.
\newblock Study and benchmarks for real root isolation methods.
\newblock Master's thesis, Dept. Informatics \& Telecoms, University of Athens,
  2009.
\newblock www.zafeirakopoulos.info/content/publications/thesis.pdf.

\end{thebibliography}
}

\end{document}